\documentclass[10pt,showkeys,showpacs,twocolumn]{revtex4-1}
\usepackage{graphicx}
\usepackage{amsmath}
\usepackage{amsfonts,amssymb}
\usepackage[dvips]{color}
\usepackage{color}

\newcommand{\pslash}{\not{\hbox{\kern-2.3pt $p$}}}
\newcommand{\pdslash}{\not{\hbox{\kern-2pt $\partial$}}}
\newtheorem{theorem}{Theorem}[section]

\newtheorem{corollary}{Corollary}[section]
\newenvironment{proof}[1][Proof]{\noindent\textbf{#1.} }{\ \rule{0.5em}{0.5em}}
\begin{document}

\title[]{An Approach by Representation of Algebras for Decoherence-Free Subspaces}%

\author{M. A. S. Trindade}%
\affiliation{Departamento de Ci\^{e}ncias Exatas e da Terra, Universidade do Estado da Bahia, Rodovia Alagoinhas/Salvador,
BR 110, Km 03, 48040-210, Alagoinhas, Bahia, Brazil.}

\author{Eric Pinto}%
\affiliation{Departamento de Matem\'atica, Universidade Federal da Bahia, Campus Ondina, 40210-340, Salvador, Bahia, Brazil.}

\author{J. D. M. Vianna}%
\affiliation{Instituto de F\'isica, Universidade Federal da Bahia, Campus Ondina, 40210-340, Salvador, Bahia, Brazil.}
\affiliation{International Center for Condensed Matter Physics, Instituto de F\'isica, Universidade de Bras\'ilia, 70910-900, Bras\'ilia, DF, Brazil.}

\author{PACS numbers: 03.65.Fd, 03.65.Yz, 03.67.-a}

\begin{abstract}
The aim of this paper is to present a general algebraic formulation for the Decoherence-Free Subspaces (DFSs). For this purpose, we initially generalize some results of Pauli and Artin about semisimple algebras. Then we derive orthogonality theorems for algebras analogous to finite groups. In order to build the DFSs we consider the tensor product of Clifford algebras and left minimal ideals. Furthermore, we show that standard applications of group theory in quantum chemistry can be obtained in our formalism. Advantages and some perspectives are also discussed.
\end{abstract}


\maketitle

\section{Introduction}
Group theoretical methods play a fundamental role in physics \cite{hamermesh,lomont,tinkhan}. In this context, many seminal works appeared in the scope of  quantum mechanics. The numbers and kinds of energy levels, for example, are determined by symmetry group of molecules. The theory of representation of finite groups enables applications in molecular vibrations, molecular orbital theory, transition metal chemistry and many others \cite{Bishop}. Infra-red spectra, ultra-violet spectra, dipoles moment and optical activities are physical properties which depend on molecular symmetry. It is shown that in the problems involving large numbers of orbitals, high-order secular equations can be formulated so that symmetry considerations simplify these equations \cite{Cotton}. In the modern theory of quantum computation, Deutsch \cite{Deu} showed that the Fourier transform over group $Z_{2^{n}}$ may be efficiently implemented in a quantum computer and many others problems can be described by the hidden subgroup problem \cite{Mosca,Jos} such as Schor's quantum algorithm for factoring and discrete logarithm, the graph isomorphism problem and certain shortest vector problems in lattices \cite{Oded}. Furthermore, in the quantum error correction,  one-dimensional representations of group algebras allow the characterization and construction of the decoherence-free subspaces (DFSs) for multiple-qubit errors \cite{lidar1,lidar2,lidar11}.

Decoherence \cite{dec1,nielsen,fortin} is responsible for the loss of quantum information from a system into the environment. Since the protection of quantum information is a central task in quantum information processing, decoherence mitigation strategies are important. One of them corresponds to decoherence-free subspaces. It is shown \cite{lidar1,lidar2} that when the Kraus operators are viewed as operators in algebra of the Pauli group, the decoherence free states belong to the one-dimensional irreducible representations of the Pauli group. In this scheme, these subspaces may be appear without spatial symmetry. A more general approach using the concept of interaction algebra was performed by Knill \emph{et al.} \cite{Knill}. Decoherence-free subspaces can be recognized as special case of general decomposition for appropriate graded interaction algebra. This is a general formalism for the theory of quantum error-correction with arbitrary noise. This approach visualizes the notion of error correcting codes as abstract particles, associating with irreducible representations of closed operator algebras and Hermitian conjugation.\

Algebraic formulations has enabled an alternative way to describe physical theories \cite{formula3,Holland,formula4,formula5,formula6,formula7,Schonberg,Schonberg1,Bohm,Vianna0,pinto1}. Dirac \cite{Dirac,Hileyn} pointed out that when algebraic methods are used for systems with an infinite numbers of degrees of freedom, it is possible to obtain solutions to some physical problems that give no solution in the usual Schr\"{o}dinger picture. Sch\"{o}nberg \cite{Schonberg} using relationships between Clifford and Grassmann algebras provided a better understanding of the relativistic phase space. He suggested that there is a deep relationship between quantum theory and geometry and this may be noticed by the observation that the formalism of quantum mechanics and quantum field theory can be interpreted as special kinds of geometric algebras \cite{formula3}. These  are algebras of symmetric and anti-symmetric tensors that have the same structure of the boson and fermion algebras related to annihilation and creation operators. In this perspective, an alternative interpretation was obtained by Bohm \cite{Bohm}, visualizing the spin as a property of a field of anti-symmetric tensor, characterizing a new type of motion. Still, in a algebraic scenario, the local quantum field theory, introduced by Haag and Kastler \cite{Haag}, is an application of C*-algebras to quantum field theory.  In this context, pure states correspond to irreducible representation, according to GNS construction. This technique capture the elements that should provide the physical foundations of a mathematically consistent formalism \cite{Emch}. In particle physics, an interesting application of Clifford algebras in flavor symmetry of quarks has been carried out with the concept of isotopy \cite{Roldao}. For the semiconductors physics, Dargys \cite{dargys1,dargys2} has  shown advantages of the application of Clifford algebras in the treatment of spin dynamics.

With regard to quantum computation, many authors \cite{Shy,Bay,Havel,Vlasov,Trindade} have investigated applications of Clifford algebras to describe quantum entanglement and numerous aspects related to this subject. The key idea is that operators and operands should be elements of same space. This can be performed using the concept of minimal left ideals or Clifford modules. By making tensor product of algebras and their minimal left ideals, it is possible \cite{Trindade} to describe states and operators in composite quantum systems within the same algebraic structure without resort to any representation on the Hilbert spaces. Consequently, an interesting aspect of this formulation is that both states and operators can be represented by means of the generators of the algebra, which is an advantage from the operational point of view.

In this paper we present a general algebraic formulation for description of quantum systems using orthogonality theorems for algebras and we show how these results can be useful for the construction of the DFSs. The paper is unfolded in the following sequence of presentation. In the Section \ref{sec:representation1} we developed some general results about irreducible representations of algebras. Section \ref{sec:results3} is devoted to applications in the algebraic construction of decoherence-free subspaces. In Sec. \ref{sec:conclu}, we present the conclusions.

\section{Representation Theory of Algebras} \label{sec:representation1}

We will start this section stating three theorems of representation theory of algebras. The first two were demonstrated by Pauli and Artin, and the third by Fr\"{o}benius and Schur \cite{pauli}. Then we will make our generalizations. This is relevant to our formulation in order that we can analyze composite systems.

\begin{theorem} \label{theo17}
Let A be an algebra and $D(A)$ a representation of A. We denote the basis elements of A by $a_{i}$ and $D(a_{i})$ its matrix
representation, where $i=1,2,\ldots,n$. Define $D_{ij}\equiv Tr[D(a_{i})D(a_{j})]$. Then A is semisimple algebra if and only if
$det\parallel D_{ij}\parallel \neq 0$.
\end{theorem}

\begin{theorem}\label{theo18}
Let A be a semisimple algebra. Then the number of irreducible representations is equal to the number of basis elements that commute
with one another.
\end{theorem}

\begin{theorem}
Let A be an algebra with dimension $n$. Consider k the number of nonequivalent irreducible representations of A and its dimensions
are respectively $n_{1}, n_{2}, \ldots, n_{k}$. Then $n=n_{1}^{2}+n_{2}^{2}+ \cdots+n_{k}^{2}$.
\end{theorem}

Generalizing the theorems of Pauli and Artin to the tensor product of algebras, we have

\begin{theorem} \label{theo20}
Let $A_{1}, A_{2},\ldots, A_{m}$ be a semisimple algebras. Then the tensor product $A_{1}\otimes A_{2}\otimes \cdots \otimes A_{m}$ is also semisimple.
\end{theorem}

\begin{proof}
Let $D(A_{1}),D(A_{2}),\ldots,D(A_{m})$ be a representations of $A_{1},A_{2},\ldots,A_{m}$, respectively. We denote the basis elements of $A_{m}$ by $a_{i_{m}}^{m}$  with $i_{m}=1,2,\ldots,n_{m}$ and $D(a_{i_{m}}^{m})$ is a matrix representation. We define the quantities as
\begin{eqnarray}
D_{i_{m}j_{m}}^{m}\equiv Tr\left[ D(a_{i_{m}}^{m})D(a_{j_{m}}^{m})\right]
\end{eqnarray}

and

\begin{eqnarray}
D_{ij}^{1,2,\ldots,m}&\equiv& Tr[ D(a_{i_{1}}^{1}\otimes
a_{i_{2}}^{2}\otimes \cdots \otimes a_{i_{m}}^{m}) \nonumber \\
&& D(a_{j_{1}}^{1}\otimes a_{j_{2}}^{2}\otimes \cdots \otimes a_{j_{m}}^{m})]
\end{eqnarray}
where the $i_{1},i_{2},...,i_{m}\longrightarrow i$ and $j_{1},j_{2},\ldots,j_{m}\longrightarrow j$ correlation is
performed according to the dictionary order \cite{hamermesh} to $i_{m}$ and $j_{m}$.

For last expression, we have
\begin{eqnarray*}
D_{ij}^{1,2,\ldots,m}&=&Tr\{[D(a_{i_{1}}^{1})D\left(a_{j_{1}}^{1}\right)] \otimes [D(a_{i_{2}}^{1}) \nonumber \\
&& D(a_{j_{2}}^{1})] \otimes \cdots \otimes [D(a_{i_{m}}^{m})D(a_{j_{m}}^{m})]\} \nonumber \\
&=&Tr\{[ D(a_{i_{1}}^{1})D(a_{j_{1}}^{1})][D(a_{i_{2}}^{1})D(a_{j_{2}}^{1})] \\
&& \cdots[D(a_{i_{m}}^{m})D( a_{j_{m}}^{m})] \} \\
&=&D_{i_{1}j_{1}}^{1}D_{i_{2}j_{2}}^{2} \cdots D_{i_{m}j_{m}}^{m}.
\end{eqnarray*}
On the other hand, we have
\begin{eqnarray*}
\det \left\Vert D_{ij}^{1,2,\ldots,m}\right\Vert&=&\det \left\Vert
D_{i_{1}j_{1}}^{1}D_{i_{2}j_{2}}^{2}\cdots D_{i_{m}j_{m}}^{m}\right\Vert  \\
&=&\det \left\Vert D_{i_{1}j_{1}}^{1}\right\Vert \det \left\Vert
D_{i_{2}j_{2}}^{2}\right\Vert \\
&& \cdots \det \left\Vert D_{i_{m}j_{m}}^{m}\right\Vert.
\end{eqnarray*}

As $A_{1},A_{2},\ldots,A_{m}$ are semisimple algebras, supposedly, therefore these determinants
$\det \left\Vert D_{i_{1}j_{1}}^{1}\right\Vert, \\
\det \left\Vert D_{i_{2}j_{2}}^{2}\right\Vert,\ldots, \det \left\Vert D_{i_{m}j_{m}}^{m}\right\Vert$ are all nonzero.
Therewith $\det \left\Vert D_{ij}^{1,2,\ldots,m}\right\Vert \neq 0$. Recalling that the Pauli-Artin's theorem tells us
that the algebra $A_{1}\otimes A_{2}\otimes \cdots \otimes A_{m}$ is semisimple if and only if
$\det \left\Vert D_{ij}^{1,2,\ldots,m}\right\Vert \neq 0$. Accordingly, the algebra $A_{1}\otimes A_{2}\otimes
\cdots \otimes A_{m}$ is semisimple.
\end{proof}

\begin{corollary}
If at least one algebra of the product $A_{1}\otimes A_{2}\otimes \cdots \otimes A_{m}$ is simple, the product
algebra $A_{1}\otimes A_{2}\otimes \cdots \otimes A_{m}$ is also simple.
\end{corollary}

\begin{theorem}
Let $n_{1},n_{2},\ldots,n_{m}$ be the number of nonequivalent irreducible representations of
algebras $A_{1},A_{2},\ldots,A_{m}$, respectively. Thus the number $M$ of nonequivalent irreducible representations of
algebra $A_{1}\otimes A_{2}\otimes \cdots \otimes A_{m}$ is given by $M=n_{1}n_{2}\ldots n_{m}$.
\end{theorem}

\begin{proof}
Let $b_{i_{1}}^{1},b_{i_{2}}^{2},\ldots,b_{i_{m}}^{m}$ be a basis elements of the algebras $A_{1},A_{2},\ldots,A_{m}$ respectively,
with $i_{1}=1,2,\ldots,r_{1};~i_{2}=1,2,\ldots,r_{2};\ldots;i_{m}=1,2,...,r_{m}$, where $r_{1},r_{2},\ldots,r_{m}$ are the dimensions
of algebras $A_{1},A_{2},\ldots,A_{m}$. Consider the elements $c_{j_{1}}^{1},c_{j_{2}}^{2},\ldots,c_{j_{m}}^{m}$ that correspond to
the basis elements which commute with all elements of the respective algebras, i.e.:
\begin{equation}
\left[ b_{i_{1}}^{1},c_{j_{1}}^{1}\right] =\left[ b_{i_{2}}^{2},c_{j_{2}}^{2}%
\right] =\cdots=\left[ b_{i_{m}}^{m},c_{j_{m}}^{m}\right] =0.
\end{equation}

An basis element in the product algebra has the form $b_{i_{1}}^{1}\otimes b_{i_{2}}^{2} \otimes \cdots \otimes b_{i_{m}}^{m}$. Let
$x_{j_{1}}^{1}\otimes x_{j_{2}}^{2}\otimes \cdots \otimes x_{j_{m}}^{m}$ be a basis elements which commute with all elements, i.e.:
\begin{equation}
\left[ b_{i_{1}}^{1}\otimes b_{i_{2}}^{2}\otimes \cdots \otimes b_{i_{m}}^{m},x_{j_{1}}^{1}\otimes x_{j_{2}}^{2}\otimes \cdots
\otimes x_{j_{m}}^{m}\right]=0
\end{equation}
Accordingly, we have
\begin{eqnarray*}
\left(b_{i_{1}}^{1}\otimes \cdots \otimes b_{i_{m}}^{m}\right)\left(x_{j_{1}}^{1} \otimes \cdots \otimes x_{j_{m}}^{m}\right)&=&
\end{eqnarray*}
\begin{eqnarray}\label{relatresnove}
x_{j_{1}}^{1}b_{i_{1}}^{1}\otimes x_{j_{2}}^{2}b_{i_{2}}^{2}\otimes
\cdots \otimes x_{j_{m}}^{m}b_{i_{m}}^{m}.
\end{eqnarray}
Hence, a condition for which the relation (\ref{relatresnove}) to be verified is that the elements $x_{j_{m}}^{m}$ commute with all elements
$b_{i_{m}}^{m}$. Consequently $x_{j_{m}}^{m}=c_{j_{m}}^{m}$. Moreover, as the number of nonequivalent irreducible representations of an algebra is the number of basis elements that commute with all others, we have $j_{1}=1,2,\ldots,n_{1};\j_{2}=1,2,\ldots,n_{2};\ldots;
\ j_{m}=1,2,\ldots,n_{m}$. The number of basis elements of a product algebra that commute with all other elements will be given by $n_{1}n_{2}...n_{m}$, corresponding to all possible combinations of the tensor products of $c_{j}$. Thus, the $M$ number of nonequivalent irreducible representations is given by $M=n_{1}n_{2}...n_{m}$.
\end{proof}

The following theorem holds a strong analogy with the orthogonality theorems of representation theory of finite groups \cite{lomont} adapted to algebras that satisfy the characteristic equation:
\begin{equation}\label{relacarasem}
a_{x}a_{y}=c_{xy}^{z}a_{z};\ \ \ \ \ \ \ \ \ x,y,z=1,2,...,d,
\end{equation}
where all basis elements are invertible. It is noteworthy that Einstein summation convention is not used. A wide class of algebras used in physics satisfies the relationship (\ref{relacarasem}). For example group algebras, Clifford algebras, quaternions and Pauli algebra.

\begin{theorem}
Let $A$ be an algebra of dimension $d$, with invertible basis elements $\left\{ a_{1},a_{2},\ldots,a_{d}\right\}$ satisfies the characteristic equation (\ref{relacarasem}). Consider all irreducible and nonequivalent unitary representations of this algebra. If $D^{\alpha }$ is $d_{\alpha}$-dimensional representation and $D^{\beta }$ is $d_{\beta}$-dimensional representation, then
\begin{equation}\label{teograorto}
\underset{y}{\sum }D_{ij}^{\ast \alpha }(a_{y})D_{kl}^{\beta }(a_{y})=\frac{d%
}{d_{\alpha }}\delta _{\alpha \beta }\delta _{ik}\delta _{jl},
\end{equation}%
where the summation is performed over all basis elements of algebra.
\end{theorem}

\begin{proof}
Let $T$ be a matrix with dimension $d_{\alpha}\times d_{\beta}$, so
\begin{equation}\label{Aescolha}
T=\underset{y}{\sum }D^{\beta }(a_{y})BD^{\dagger \alpha }(a_{y}),
\end{equation}%
where the matrix $B$ of dimension $d_{\beta }\times d_{\alpha }$ it is completely arbitrary. Now multiply the left side by $D^{\beta }(a_{x})$, we obtain
\begin{eqnarray}\label{Aes1}
D^{\beta }(a_{x})T &=&\underset{y}{\sum }D^{\beta }(a_{x})D^{\beta}(a_{y})BD^{\dagger \alpha }(a_{y}) \nonumber \\
&=&\underset{y}{\sum }D^{\beta }(a_{x})D^{\beta }(a_{y})BD^{\dagger \alpha}(a_{y})D^{\dagger \alpha }(a_{x})D^{\alpha }(a_{x}) \nonumber \\
&=&\underset{y}{\sum }D^{\beta }(a_{x}a_{y})BD^{\dagger \alpha}(a_{x}a_{y})D^{\alpha }(a_{x}).
\end{eqnarray}
Since $y$ subindex varies over all basis elements, $x$ is fixed and $z$ changes, we will rewrite the equation (\ref{relacarasem}). Then
\begin{equation}\label{relacarasem2}
a_{x}a_{y}=c_{xy}^{z}a_{z_{y}}.
\end{equation}%
Substituting (\ref{relacarasem2}) in Eq. (\ref{Aes1}), we find
\begin{eqnarray*}
D^{\beta }(a_{x})T &=&\underset{y}{\sum }D^{\beta
}(c_{xy}^{z}a_{z_{y}})BD^{\dagger \alpha }(c_{xy}^{z}a_{z_{y}})D^{\alpha}(a_{x}) \nonumber \\
&=&\underset{y}{\sum }c_{xy}^{z}D^{\beta }(a_{z_{y}})B(c_{xy}^{z})^{\ast
}D^{\dagger \alpha}(a_{z_{y}})D^{\alpha }(a_{x}).
\end{eqnarray*}%
The unitary representations are
\begin{eqnarray}
D(c_{xy}^{z}a_{z_{y}})D^{\dagger }(c_{xy}^{z}a_{z_{y}}) &=&1,
\notag \\
c_{xy}^{z}D(a_{z_{y}})(c_{xy}^{z})^{\ast }D^{\dagger }(a_{z_{y}}) &=&1,
\notag \\
c_{xy}^{z}(c_{xy}^{z})^{\ast }D(a_{z_{y}})D^{\dagger }(a_{z_{y}}) &=&1.
\end{eqnarray}%
Note that $D(a_{z_{y}})D^{\dagger }(a_{z_{y}})=1$. Consequently
\begin{equation}
c_{xy}^{z}(c_{xy}^{z})^{\ast }=1.
\end{equation}%
so
\begin{equation*}
D^{\beta }(a_{x})T=\underset{y}{\sum }D^{\beta }(a_{z_{y}})BD^{\dagger
\alpha }(a_{z_{y}})D^{\alpha }(a_{x}).
\end{equation*}%
We also note that for each $y$ have a distinct $z$ (if they were equal, the elements of algebra would be linearly dependent, a contradiction). Therefore we can write
\begin{eqnarray*}
D^{\beta }(a_{x})T &=&\underset{y}{\sum }D^{\beta }(a_{y})BD^{\dagger \alpha
}(a_{y})D^{\alpha }(a_{x})  \notag \\
D^{\beta }(a_{x})T &=&TD^{\alpha }(a_{x}),
\end{eqnarray*}%
with all $x\in 1,...,d$. Hence by Schur's lemma $T=0$, when $\alpha
\neq \beta $, i.e, $D^{\alpha }$ and $D^{\beta }$ are not equivalents. The element $A_{ki}$  is
\begin{equation}
\underset{y}{\sum }\underset{m,n}{\sum }D_{km}^{\beta
}(a_{y})B_{mn}D_{ni}^{\alpha }(a_{y}^{-1})=0,
\end{equation}%
where the unitarity of $D(a_{y})$ was used. Since $B$ is a completely arbitrary matrix, we will make $B_{lj}=1$ and the rest of its elements equal to zero, the last equation is
\begin{equation}
\underset{y}{\sum }D_{kl}^{\beta }(a_{y})D_{ji}^{\alpha }(a_{y}^{-1})=0
\end{equation}
from unitarity of $D(a_{y})$
\begin{equation}\label{grandzero}
\underset{y}{\sum }D_{ij}^{\ast \alpha }(a_{y})D_{kl}^{\beta }(a_{y})=0.
\end{equation}%
However, when $\alpha =\beta $, i.e., $D^{\alpha }$ and $D^{\beta }$ are equivalents, $T$ is a scalar matrix by Schur's lemma.
For $T=a\delta _{ki}$, from (\ref{Aescolha}), we obtain
\begin{eqnarray}
a\delta _{ki} &=&\underset{y}{\sum }\underset{m,n}{\sum }D_{km}^{\alpha
}(a_{y})B_{mn}D_{ni}^{\alpha }(a_{y}^{-1}) \nonumber \\
&=&\underset{y}{\sum }D_{kl}^{\alpha }(a_{y})D_{ji}^{\alpha }(a_{y}^{-1}),
\end{eqnarray}%
Considering the case $k=i$ and summing over $i$ on both sides of the equation
\begin{eqnarray}
a\underset{i}{\sum }\delta _{ii} &=&ad_{\alpha } \nonumber \\
&=&\underset{i}{\sum }\underset{y}{\sum }D_{il}^{\alpha
}(a_{y})D_{ji}^{\alpha }(a_{y}^{-1})  \nonumber\\
&=&\underset{y}{\sum }\underset{i}{\sum }D_{ji}^{\alpha
}(a_{y}^{-1})D_{il}^{\alpha }(a_{y})  \nonumber \\
&=&\underset{y}{\sum }D_{jl}^{\alpha }(a_{y}^{-1}a_{y}) \nonumber \\
&=&\underset{y}{\sum }D_{jl}^{\alpha }(e)=d,
\end{eqnarray}%
from which it follows that $a=d/d_{\alpha}$. Thus,
\begin{equation}\label{grandefim}
\underset{y}{\sum }D_{ij}^{\ast \alpha }(a_{y})D_{kl}^{\alpha }(a_{y})=\frac{d}{d_{\alpha }}\delta _{ik}\delta _{jl}.
\end{equation}%
Using (\ref{grandzero}) and (\ref{grandefim}), we find the result
\begin{equation}
\underset{y}{\sum }D_{ij}^{\ast \alpha }(a_{y})D_{kl}^{\beta }(a_{y})=\frac{d%
}{d_{\alpha }}\delta _{\alpha \beta }\delta _{ik}\delta _{jl},
\end{equation}
\end{proof}

Now we find orthogonality relations for more general situations than those presented in the previous theorem. The following theorem shows these conditions.
\begin{theorem} \label{theoII7}
Let $A$ be an algebra of dimension $d$ with invertible basis elements $\left\{ a_{1},a_{2},...,a_{d}\right\}$ satisfies the following relationship:
\begin{equation}
a_{x}a_{y}=\underset{z=1}{\overset{n}{\sum }}c_{xy}^{z}a_{z},
\end{equation}%
where $\underset{y=1}{\overset{n}{\sum }}$ $c_{xy}^{z}c_{xy}^{z^{\prime
}\ast }=\delta _{zz^{\prime }}$ is valid for all $x$. Consider all irreducible nonequivalent unitary representations of this algebra.
If $D^{\alpha }$ and $D^{\beta }$ are two representations with dimensions $d_{\alpha}$ and $d_{\beta }$ respectively, so
\begin{equation}
\underset{y}{\sum }D_{ij}^{\ast \alpha }(a_{y})D_{kl}^{\beta }(a_{y})=\frac{d%
}{d_{\alpha }}\delta _{\alpha \beta }\delta _{ik}\delta _{jl},
\end{equation}%
where the sum is performed over all basis elements of algebra.
\end{theorem}

\begin{proof}
Let $T$ be a $d_{\alpha }\times d_{\beta }$-dimensional matrix. Then
\begin{equation}
T=\underset{y}{\sum }D^{\beta }(a_{y})BD^{\dagger \alpha }(a_{y})
\end{equation}
where $B$ is an arbitrary matrix of dimension $d_{\beta }\times d_{\alpha }$ and the sum is performed over all basis elements of algebra. Multiply by $D^{\beta }(a_{x})$, we have
\begin{eqnarray}
D^{\beta }(a_{x})T &=&\underset{y}{\sum }D^{\beta }(a_{x})D^{\beta
}(a_{y})BD^{\dagger \alpha }(a_{y}) \notag \\
&=&\underset{y}{\sum }D^{\beta }(a_{x})D^{\beta }(a_{y})BD^{\dagger \alpha
}(a_{y})\nonumber \\
&&D^{\dagger \alpha }(a_{x})D^{\alpha }(a_{x})  \notag \\
&=&\underset{y}{\sum }D^{\beta }(a_{x}a_{y})BD^{\dagger \alpha
}(a_{x}a_{y})D^{\alpha }(a_{x})  \notag \\
&=&\underset{y}{\sum }D^{\beta }(\underset{z}{\sum }c_{xy}^{z}a_{z})BD^{%
\dagger \alpha }(\underset{z^{\prime }}{\sum }c_{xy}^{z^{\prime
}}a_{z^{\prime }})\nonumber \\
&&D^{\alpha }(a_{x})  \notag \\
&=&\underset{y}{\sum }\underset{z,z^{\prime }}{\sum }c_{xy}^{z}c_{xy}^{z^{%
\prime }\ast }D^{\beta }(a_{z})BD^{\dagger \alpha }(a_{z^{\prime
}})\nonumber \\
&&D^{\alpha }(a_{x})  \notag \\
&=&\underset{zz^{\prime }}{\sum }\left( \underset{y}{\sum }%
c_{xy}^{z}c_{xy}^{z^{\prime }\ast }\right) D^{\beta }(a_{z})BD^{\dagger
\alpha }(a_{z^{\prime }})\nonumber \\
&&D^{\alpha }(a_{x})  \notag \\
&=&\underset{zz^{\prime }}{\sum }\delta _{zz^{\prime }}D^{\beta
}(a_{z})BD^{\dagger \alpha }(a_{z^{\prime }})D^{\alpha }(a_{x})  \notag \\
&=&\underset{z^{\prime }}{\sum }D^{\beta }(a_{z})BD^{\dagger \alpha
}(a_{z^{\prime }})D^{\alpha }(a_{x}).
\end{eqnarray}%
From now on the proof is completely analogous to the previous theorem.
\end{proof}

This condition on the characteristic equation of algebra also ensures the existence of an unitary equivalent representation to any representation of the algebra. The following theorem guarantees such assertion.

\begin{theorem} \label{theo24}
Every representation of a finite-dimensional associative algebra $A$ in which all elements are invertible and the coefficients of the characteristic equation satisfy the relation
\begin{equation}
\underset{y=1}{\overset{n}{\sum }}c_{xy}^{z}c_{xy}^{z^{\prime }\ast }=\delta
_{zz^{\prime }}
\end{equation}
it is equivalent to a unitary representation.
\end{theorem}

\begin{proof}
Let $\left\{ a_{1},a_{2},...,a_{d}\right\} $ be basis elements of the algebra and consider the following hermitian matrix
\begin{equation}
F=\underset{y}{\sum }D(a_{y})D^{\dagger }(a_{y}),
\end{equation}%
where the sum is performed over all basis elements. All Hermitian matrix can be diagonalized by a unitary matrix $U$. Then
\begin{eqnarray}
\Lambda  &=&U\underset{y}{\sum }D(a_{y})D^{\dagger }(a_{y})U^{-1} \notag \\
&=&\underset{y}{\sum }UD(a_{y})U^{-1}UD^{\dagger }(a_{y})U^{-1}  \notag \\
&=&\underset{y}{\sum }\left[ UD(a_{y})U^{-1}\right] \left[ UD^{\dagger
}(a_{y})U^{-1}\right]   \notag \\
&=&\underset{y}{\sum }\left[ UD(a_{y})U^{-1}\right] \left[
(U^{-1})D^{\dagger }(a_{y})U^{\dagger }\right]   \notag \\
&=&\underset{y}{\sum }\left[ UD(a_{y})U^{-1}\right] \left[ UD(a_{y})U^{-1}%
\right] ^{\dagger }.
\end{eqnarray}%
Define $T(a_{y})=UD(a_{y})U^{-1}$, we obtain
\begin{equation}
\Lambda =\underset{y}{\sum }T(a_{y})T^{\dagger }(a_{y}).
\end{equation}%
The set $\left\{T(a_{y})\right\}$ form an equivalent representation of the matrix $\left\{ D(a_{y})\right\} $. Since
\begin{equation*}
\Lambda ^{-1/2}\Lambda \Lambda ^{-1/2}=I
\end{equation*}%
and
\begin{equation*}
\left( \Lambda ^{1/2}\right) ^{\dagger }=\Lambda ^{1/2},
\end{equation*}%
we have
\begin{equation}\label{quatrinta}
\Lambda ^{-1/2}\underset{y}{\sum }T(a_{y})T^{\dagger }(a_{y})\Lambda
^{-1/2}=I,
\end{equation}%
where $\Lambda^{\pm 1/2}$  is obtained taking the square root $(+)$ or the inverse of the square root $(-)$ of the elements of $\Lambda$. Let us define the matrix
\begin{equation}
R(a_{x})=\Lambda ^{-1/2}T(a_{x})\Lambda ^{1/2}
\end{equation}
\newline
This matrix is unitary and
\begin{eqnarray*}
R(a_{x})R^{\dagger }(a_{x})=\Lambda ^{-1/2}T(a_{x})\Lambda ^{1/2}\Lambda
^{1/2}T^{\dagger }(a_{x})\Lambda ^{-1/2}
\end{eqnarray*}
\begin{eqnarray*}
&=&\Lambda ^{-1/2}T(a_{x})\Lambda ^{1/2}(\Lambda ^{-1/2}\underset{y}{\sum }T(a_{y})T^{\dagger }(a_{y}) \nonumber \\
&&\Lambda ^{-1/2})\Lambda ^{1/2}T^{\dagger}(a_{x})\Lambda ^{-1/2}  
\end{eqnarray*}
\begin{eqnarray*}
&=&\Lambda ^{-1/2}\underset{y}{\sum }\left[ T(a_{x})T(a_{y})\right] \left[
T(a_{x})T(a_{y})\right] ^{\dagger }\Lambda ^{-1/2}  \nonumber \\
&=&\Lambda ^{-1/2}\underset{y}{\sum }T(\underset{z}{\sum }%
c_{xy}^{z}a_{z})T^{\dagger }(\underset{z^{\prime }}{\sum }c_{xy}^{z^{\prime
}}a_{z^{\prime }})\Lambda ^{-1/2}  \nonumber \\
&=&\Lambda ^{-1/2}\underset{zz^{\prime }}{\sum }\underset{y}{\sum }%
c_{xy}^{z}c_{xy}^{z^{\prime }\ast }T(a_{z})T^{\dagger }(a_{z^{\prime
}})\Lambda ^{-1/2}  \nonumber \\
&=&\Lambda ^{-1/2}\underset{zz^{\prime }}{\sum }\delta _{zz^{\prime
}}T(a_{z})T^{\dagger }(a_{z^{\prime }})\Lambda ^{-1/2}  \nonumber \\
&=&\Lambda ^{-1/2}\underset{z}{\sum }T(a_{z})T^{\dagger }(a_{z})\Lambda
^{-1/2}=I.
\end{eqnarray*}%
From equation (\ref{quatrinta}), we have
\begin{equation*}
R(a_{x})R^{\dagger }(a_{x})=\Lambda ^{-1/2}\underset{z}{\sum }%
T(a_{z})T^{\dagger }(a_{z})\Lambda ^{-1/2}=I.
\end{equation*}%
Using the unitary matrix $U$ that diagonalizes $F$, the matrices $\left\{ D(a_{y})\right\}$ are
transformed into unitary matrices through of the matrices $\Lambda ^{-1/2}U$, i.e.
\begin{equation*}
R(a)=\left( \Lambda ^{-1/2}U\right) D(a)\left( \Lambda ^{-1/2}U\right) ^{-1}.
\end{equation*}%
Therefore, the matrices $\left\{ D(a)\right\} $ are always convertible into unitary matrices.
\end{proof}


\begin{theorem}\label{theoremIII.9}
Let $A$ be an algebra of dimension $d$ that satisfies the characteristic equation defined by the previous theorem, in other words
$\underset{y=1}{\overset{n}{\sum }}$ $c_{xy}^{z}c_{xy}^{z^{\prime }\ast }=\delta_{zz^{\prime }}$. The characteres $\chi ^{\alpha }(a_{y})$
and $\chi ^{\beta}(a_{y})$ associated with two irreducible representations $\Gamma ^{\alpha }$ and $\Gamma ^{\beta }$, respectively, satisfy the orthogonality relation
\begin{equation} \label{cara}
\underset{y}{\sum }\chi ^{\ast \alpha }(a_{y})\chi ^{\beta }(a_{y})=d\delta
_{\alpha \beta }.
\end{equation}
\end{theorem}

\begin{proof}
From the previous theorem,
\begin{equation*}
\underset{y}{\sum }D_{ij}^{\ast \alpha }(a_{y})D_{kl}^{\beta }(a_{y})=\frac{d%
}{d_{\alpha }}\delta _{\alpha \beta }\delta _{ik}\delta _{jl}.
\end{equation*}%
and setting $i=j,k=l$, we have
\begin{eqnarray*}
\underset{i}{\sum }\underset{k}{\sum }\left( \underset{y}{\sum }D_{ii}^{\ast
\alpha }(a_{y})D_{kk}^{\ast \beta }(a_{y})\right)=
\end{eqnarray*}
\begin{eqnarray*}
&&\underset{y}{\sum }%
\left( \underset{i}{\sum }D_{ii}^{\ast \alpha }(a_{y})\right) \left(
\underset{k}{\sum }D_{kk}^{\ast \beta }(a_{y})\right) \notag \\
&=&\underset{y}{\sum }D_{ij}^{\ast \alpha }(a_{y})D_{kl}^{\ast \beta
}(a_{y})=\frac{d}{d\alpha }\delta _{\alpha \beta }\underset{i}{\sum }%
\underset{k}{\sum }(\delta _{ik})^{2}  \notag \\
&=&d\delta _{\alpha \beta}
\end{eqnarray*}%
which proves our theorem.
\end{proof}

\subsection{The Projection Operators}

In this subsection, we will construct projection operators \cite{hamermesh} in our formulation. The result of the operation of any base element algebra on the function $\varphi_{k}^{(j)}$ is expressible as
\begin{equation}
P_{a_{y}}\varphi _{k}^{(j)}=\underset{\lambda =1}{\overset{d_{j}}{\sum }}%
\varphi _{\lambda }^{(j)}D_{\lambda k}^{(j)}(a_{y}),
\end{equation}%
where $d_{j}$ is the dimension of representation $\Gamma ^{j}$. If we multiply this expression by $D_{\lambda ^{^{\prime }}k^{^{\prime
}}}^{(i)}(a_{y})$ and we add on the basis elements of algebra, we obtain
\begin{eqnarray}\label{opproj2}
\underset{y}{\sum }D_{\lambda ^{^{\prime }}k^{^{\prime
}}}^{(i)}(a_{y})P_{a_{y}}\varphi _{k}^{(j)} &=&\underset{\lambda =1}{\overset%
{d_{j}}{\sum }}\underset{y}{\sum }D_{\lambda ^{^{\prime }}k^{^{\prime
}}}^{(i)}(a_{y})\nonumber \\
&&D_{\lambda k}^{(j)}(a_{y})\varphi _{\lambda }^{(j)} \notag \\
&=&\frac{d}{d_{j}}\delta
_{ij}\delta _{kk^{^{\prime }}}\delta _{\lambda \lambda ^{^{\prime }}}\varphi
_{\lambda }^{(j)}  \notag \\
&=&\frac{d}{d_{j}}\delta
_{ij}\delta _{kk^{^{\prime }}}\varphi _{\lambda ^{^{\prime }}}^{(j)},
\end{eqnarray}%
from last Theorem \ref{theoremIII.9}. We define
\begin{equation}
P_{\lambda k}^{(j)}=\frac{d_{j}}{d}\underset{y}{\sum }D_{\lambda k}^{\ast
(j)}(a_{y})P_{a_{y}},
\end{equation}%
It is easy to see that
\begin{eqnarray}
P_{\lambda k}^{(j)}\varphi _{l}^{(i)} =\varphi _{\lambda }^{(j)}\delta
_{ij}\delta_{kl}.
\end{eqnarray}
If we apply the operator $P_{kk}^{(j)}$ to function $\phi=\underset{j^{^{\prime }}=1}{\overset{m}{\sum }}\overset{d_{j}}{%
\underset{k^{^{\prime }}=1}{\sum }}\varphi _{k^{^{\prime }}}^{(j^{^{\prime
}})}$, we obtain
\begin{eqnarray}
P_{kk}^{(j)}\phi=\varphi _{k}^{(j)}.
\end{eqnarray}
Thus the operator $P_{kk}^{(j)}$ projects only part of the function that belongs to the $k$th row of the irreducible
representation $\Gamma ^{j}$ of the space operators $P_{a_{y}}$; the operator $P_{kk}^{(j)}$ is a projection operator.
In particular, for $\lambda =k$, and summing over $k$, we have
\begin{eqnarray}
P^{(j)} &=&\underset{k}{\sum }P_{kk}^{(j)} \nonumber \\
&=&\frac{d_{j}}{d}\underset{y}{\sum }\left( \underset{k}{\sum }%
D_{kk}^{\ast (j)}(a_{y})\right) P_{a_{y}}  \nonumber \\
&=&\frac{d_{j}}{d}\underset{y}{\sum }\chi ^{\ast (j)}(a_{y})P_{a_{y}}.
\end{eqnarray}%
Applying this operator on a $\phi$ function
\begin{equation}
P^{(j)}\phi =\underset{k}{\sum }P_{kk}^{(j)}\phi =\underset{k}{\sum }\phi
_{k}^{(j)}=\phi ^{(j)},
\end{equation}%
where $\phi ^{(j)}$ is any function which can be expressed as a sum of functions belonging to the lines into the
representation $j$ and it satisfies
\begin{eqnarray}
P^{(j)}\phi ^{(j)}=\phi ^{(j)}.
\end{eqnarray}
Therefore, we can verify that
\begin{equation*}
\underset{j}{\sum }P^{(j)}=P_{e},
\end{equation*}
where $e$ is identity. Importantly, this deduction is completely analogous to the group theory approach due to our orthogonality theorems for algebras previously deducted.

\section{Algebraic Construction of Decoherence-Free Subspaces}  \label{sec:results3}

Let $H_{s}$ be Hilbert space associated with the system $S$. A DFS is a subspace $H_{s}^{\prime}=Span\{\left\vert j^{\prime}\right\rangle \}$ of the full system Hilbert space $H_{s}$ over which the evolution of the density
matrix is unitary. Lidar \textit{et al.} \cite{lidar1} showed that the decoherence-free
states ${\vert j^{\prime}\rangle }$ belongs to the one-dimensional irreducible representations of the Pauli group when the Kraus operators are characterized as operators in the algebra of the Pauli group. Since the matrices Pauli provide representations of the Clifford algebras, we can build DFSs in terms of these algebras using the theorems of the previous section.

For we use our formalism, let us consider the Clifford algebra $Cl_{3}$ over the complex field $(\mathbb{C})$
defined \cite{boerner} by
\begin{equation} \label{clifford1}
\left\{
\begin{array}{c}
\gamma _{i}^{2}=1,\ \ \ \ \ \ \ \ \ i=1,2,3 \\
\gamma _{i}\gamma _{j}+\gamma _{j}\gamma _{i}=0,\ \ \ i\neq j%
\end{array}%
\right.
\end{equation}

We can write a qubit $\left\vert \psi\right\rangle=a\left\vert 0\right\rangle+b\left\vert 1\right\rangle$ (with $\mid a\mid^{2}+\mid b\mid^{2}=1$) as the following element $\Psi$ of the minimal left ideal \cite{Trindade}, i.e.:
\begin{equation*}
\Psi=a\gamma_{1}\varepsilon_{1}+b\gamma_{3}\varepsilon_{1}
\end{equation*}
where a,b $\in \mathbb{C}$ and $\varepsilon_{1}$ is a primitive idempotent and it is given by:
\begin{equation}\label{clifford2}
\varepsilon_{1}=\frac{1}{2}(1+\gamma_{3}).
\end{equation}

An representation of this algebra is given by Pauli matrices
\begin{eqnarray}
\gamma _{1} &\Longleftrightarrow &\sigma _{1}=\left(
\begin{array}{cc}
0 & 1 \\
1 & 0%
\end{array}%
\right),
\end{eqnarray}
\begin{eqnarray}
\gamma _{2}\Longleftrightarrow \ \sigma _{2}=\left(
\begin{array}{cc}
0 & -i \\
i & 0%
\end{array}%
\right),
\end{eqnarray}
\begin{eqnarray}
\gamma _{3} &\Longleftrightarrow &\ \sigma _{3}=\left(
\begin{array}{cc}
1 & 0 \\
0 & -1%
\end{array}%
\right) .
\end{eqnarray}

and an arbitrary element of the algebra can be written as
\begin{eqnarray}
\Gamma &=&\alpha _{0}1+\alpha _{1}\gamma _{1}+\alpha _{2}\gamma _{2}+\alpha
_{3}\gamma _{3}+\alpha _{12}\gamma _{1}\gamma _{2}+ \nonumber \\
&&\alpha _{23}\gamma_{2}\gamma _{3}+\alpha _{13}\gamma _{3}\gamma _{1}+\alpha _{123}\gamma
_{1}\gamma _{2}\gamma _{3}.
\end{eqnarray}

It is easy to see that we have two elements that commute with all other
\begin{equation}
\left\{ 1,\gamma _{1}\gamma _{2}\gamma _{3}\right\} .
\end{equation}

Motivated by Theorem \ref{theo18}, we have two nonequivalent irreducible representations. A second nonequivalent irreducible
representation is given by
\begin{eqnarray}
\gamma _{1} &\Longleftrightarrow &-\sigma _{1}=\left(
\begin{array}{cc}
0 & -1 \\
-1 & 0%
\end{array}%
\right),
\end{eqnarray}
\begin{eqnarray}
\gamma _{2}\Longleftrightarrow -\ \sigma _{2}=\left(
\begin{array}{cc}
0 & i \\
-i & 0%
\end{array}%
\right),
\end{eqnarray}
\begin{eqnarray}
\gamma _{3} &\Longleftrightarrow &-\ \sigma _{3}=\left(
\begin{array}{cc}
-1 & 0 \\
0 & 1%
\end{array}%
\right) .
\end{eqnarray}

Note that this is not equivalent to the first, because the first element $i\gamma _{1}\gamma _{2}\gamma _{3}$ is mapped into
$-I$ and the second in $I$. Thus,
\begin{eqnarray}
i\gamma _{1}\gamma _{2}\gamma _{3} &\rightarrow &i\sigma _{1}\sigma
_{2}\sigma _{3}=iiI=-I,  \notag \\
i\gamma _{1}\gamma _{2}\gamma _{3} &\rightarrow &i(-\sigma _{1})(-\sigma
_{2})(-\sigma _{3})=I.
\end{eqnarray}

This information, together with the previous theorems allows us to construct nonequivalent irreducible representations for $Cl_{3}\otimes Cl_{3}$. The number $M$ of nonequivalent irreducible representations is given by $M=n_{1}n_{2}=2\times 2=4$, according to the Theorem \ref{theo20}. The four basis elements belonging to $Cl_{3}\otimes Cl_{3}$ which commute with every other elements corresponding
\begin{equation}
\left\{
\begin{array}{c}
\gamma _{1}\gamma _{2}\gamma _{3}\otimes \gamma _{1}\gamma _{2}\gamma _{3},
\\
\gamma _{1}\gamma _{2}\gamma _{3}\otimes 1, \\
1\otimes \gamma _{1}\gamma _{2}\gamma _{3}, \\
1\otimes 1.%
\end{array}%
\right.
\end{equation}

The four nonequivalent irreducible representations may be obtained through the following combinations
\begin{eqnarray*}
D^{(1)} &:&\left\{
\begin{array}{c}
\gamma _{1}^{1}\rightarrow \sigma _{1}\ \ \ \ \gamma _{2}^{1}\rightarrow
\sigma _{2}\ \ \ \ \gamma _{3}^{1}\rightarrow \sigma _{3} \\
\gamma _{1}^{2}\rightarrow \sigma _{1}\ \ \ \ \gamma _{2}^{2}\rightarrow
\sigma _{2}\ \ \ \ \gamma _{3}^{2}\rightarrow \sigma _{3}%
\end{array}%
\right.   \notag \\
D^{(2)} &:&\left\{
\begin{array}{c}
\gamma _{1}^{1}\rightarrow \sigma _{1}\ \ \ \ \gamma _{2}^{1}\rightarrow
\sigma _{2}\ \ \ \ \gamma _{3}^{1}\rightarrow \sigma _{3} \\
\gamma _{1}^{2}\rightarrow -\sigma _{1}\ \ \ \ \gamma _{2}^{2}\rightarrow
-\sigma _{2}\ \ \ \ \gamma _{3}^{2}\rightarrow -\sigma _{3}%
\end{array}%
\right.   \notag \\
D^{(3)} &:&\left\{
\begin{array}{c}
\gamma _{1}^{1}\rightarrow -\sigma _{1}\ \ \ \ \gamma _{2}^{1}\rightarrow
-\sigma _{2}\ \ \ \ \gamma _{3}^{1}\rightarrow -\sigma _{3} \\
\gamma _{1}^{2}\rightarrow \sigma _{1}\ \ \ \ \gamma _{2}^{2}\rightarrow
\sigma _{2}\ \ \ \ \gamma _{3}^{2}\rightarrow \sigma _{3}%
\end{array}%
\right.   \notag \\
D^{(4)} &:&\left\{
\begin{array}{c}
\gamma _{1}^{1}\rightarrow -\sigma _{1}\ \ \ \ \gamma _{2}^{1}\rightarrow
-\sigma _{2}\ \ \ \ \gamma _{3}^{1}\rightarrow -\sigma _{3} \\
\gamma _{1}^{2}\rightarrow -\sigma _{1}\ \ \ \ \gamma _{2}^{2}\rightarrow
-\sigma _{2}\ \ \ \ \gamma _{3}^{2}\rightarrow -\sigma _{3}%
\end{array}%
\right.
\end{eqnarray*}%
We can verify that these representations are nonequivalent if we observe the elements $i\gamma _{1}\gamma _{2}\gamma
_{3}\otimes 1$ and $1\otimes i\gamma _{1}\gamma _{2}\gamma _{3}$:
\begin{eqnarray*}
D^{(1)} &:&\left\{
\begin{array}{c}
i\gamma _{1}\gamma _{2}\gamma _{3}\otimes 1\rightarrow i\sigma _{1}\sigma
_{2}\sigma _{3}\otimes 1=-1\otimes 1 \\
1\otimes i\gamma _{1}\gamma _{2}\gamma _{3}\rightarrow I\otimes i\sigma
_{1}\sigma _{2}\sigma _{3}=-1\otimes 1%
\end{array}%
\right.   \notag \\
D^{(2)} &:&\left\{
\begin{array}{c}
i\gamma _{1}\gamma _{2}\gamma _{3}\otimes 1\rightarrow i\sigma _{1}\sigma
_{2}\sigma _{3}\otimes I=-I\otimes I \\
1\otimes i\gamma _{1}\gamma _{2}\gamma _{3}\rightarrow I\otimes i(-\sigma
_{1})(-\sigma _{2})(-\sigma _{3})=I\otimes I%
\end{array}%
\right.   \notag \\
D^{(3)} &:&\left\{
\begin{array}{c}
i\gamma _{1}\gamma _{2}\gamma _{3}\otimes 1\rightarrow i(-\sigma
_{1})(-\sigma _{2})(-\sigma _{3})\otimes I=I\otimes I \\
1\otimes i\gamma _{1}\gamma _{2}\gamma _{3}\rightarrow I\otimes i\sigma
_{1}\sigma _{2}\sigma _{3}=-I\otimes I%
\end{array}%
\right.   \notag \\
D^{(4)} &:&\left\{
\begin{array}{c}
i\gamma _{1}\gamma _{2}\gamma _{3}\otimes 1\rightarrow i(-\sigma
_{1})(-\sigma _{2})(-\sigma _{3})\otimes I=I\otimes I \\
1\otimes i\gamma _{1}\gamma _{2}\gamma _{3}\rightarrow I\otimes i(-\sigma
_{1})(-\sigma _{2})(-\sigma _{3})=I\otimes I%
\end{array}%
\right.
\end{eqnarray*}%
We have seen that given a representation where $i\gamma _{1}\gamma _{2}\gamma _{3}\otimes 1$ or $1\otimes i\gamma _{1}\gamma _{2}\gamma _{3}$ is mapped into $I\otimes I$, the other representation, it is mapped in $-I\otimes I$. This method can be generalized into a tensor product of algebras with any number of factors. For example to $Cl_{3} \otimes Cl_{3} \otimes Cl_{3}$, we have eight nonequivalent irreducible representations $\left( M=n_{1}n_{2}n_{3}=8\right)$; that is, we have eight basis elements that commute with all other elements:
\begin{equation}
\left\{
\begin{array}{c}
\gamma _{1}\gamma _{2}\gamma _{3}\otimes \gamma _{1}\gamma _{2}\gamma
_{3}\otimes \gamma _{1}\gamma _{2}\gamma _{3} \\
\gamma _{1}\gamma _{2}\gamma _{3}\otimes \gamma _{1}\gamma _{2}\gamma
_{3}\otimes 1 \\
\gamma _{1}\gamma _{2}\gamma _{3}\otimes 1\otimes \gamma _{1}\gamma
_{2}\gamma _{3} \\
\gamma _{1}\gamma _{2}\gamma _{3}\otimes 1\otimes 1 \\
1\otimes 1\otimes \gamma _{1}\gamma _{2}\gamma _{3} \\
1\otimes \gamma _{1}\gamma _{2}\gamma _{3}\otimes 1 \\
1\otimes \gamma _{1}\gamma _{2}\gamma _{3}\otimes \gamma _{1}\gamma
_{2}\gamma _{3} \\
1\otimes 1\otimes 1%
\end{array}%
\right.
\end{equation}%
The representations are given by
\begin{equation}
D^{(1)}:\left\{
\begin{array}{c}
\gamma _{1}^{1}\rightarrow \sigma _{1}\ \ \ \ \gamma _{2}^{1}\rightarrow
\sigma _{2}\ \ \ \ \gamma _{3}^{1}\rightarrow \sigma _{3} \\
\gamma _{1}^{2}\rightarrow \sigma _{1}\ \ \ \ \gamma _{2}^{2}\rightarrow
\sigma _{2}\ \ \ \ \gamma _{3}^{2}\rightarrow \sigma _{3} \\
\gamma _{1}^{3}\rightarrow \sigma _{1}\ \ \ \ \gamma _{2}^{3}\rightarrow
\sigma _{2}\ \ \ \ \gamma _{3}^{3}\rightarrow \sigma _{3}%
\end{array}%
\right.
\end{equation}%
\begin{equation}
D^{(2)}:\left\{
\begin{array}{c}
\gamma _{1}^{1}\rightarrow \sigma _{1}\ \ \ \ \gamma _{2}^{1}\rightarrow
\sigma _{2}\ \ \ \ \gamma _{3}^{1}\rightarrow \sigma _{3} \\
\gamma _{1}^{2}\rightarrow \sigma _{1}\ \ \ \ \gamma _{2}^{2}\rightarrow
\sigma _{2}\ \ \ \ \gamma _{3}^{2}\rightarrow \sigma _{3} \\
\gamma _{1}^{3}\rightarrow -\sigma _{1}\ \ \ \ \gamma _{2}^{3}\rightarrow
-\sigma _{2}\ \ \ \ \gamma _{3}^{3}\rightarrow -\sigma _{3}%
\end{array}%
\right.
\end{equation}%
\begin{equation}
D^{(3)}:\left\{
\begin{array}{c}
\gamma _{1}^{1}\rightarrow \sigma _{1}\ \ \ \ \gamma _{2}^{1}\rightarrow
\sigma _{2}\ \ \ \ \gamma _{3}^{1}\rightarrow \sigma _{3} \\
\gamma _{1}^{2}\rightarrow -\sigma _{1}\ \ \ \ \gamma _{2}^{2}\rightarrow
-\sigma _{2}\ \ \ \ \gamma _{3}^{2}\rightarrow -\sigma _{3} \\
\gamma _{1}^{3}\rightarrow \sigma _{1}\ \ \ \ \gamma _{2}^{3}\rightarrow
\sigma _{2}\ \ \ \ \gamma _{3}^{3}\rightarrow \sigma _{3}%
\end{array}%
\right.
\end{equation}%
\begin{equation}
D^{(4)}:\left\{
\begin{array}{c}
\gamma _{1}^{1}\rightarrow -\sigma _{1}\ \ \ \ \gamma _{2}^{1}\rightarrow
-\sigma _{2}\ \ \ \ \gamma _{3}^{1}\rightarrow -\sigma _{3} \\
\gamma _{1}^{2}\rightarrow \sigma _{1}\ \ \ \ \gamma _{2}^{2}\rightarrow
\sigma _{2}\ \ \ \ \gamma _{3}^{2}\rightarrow \sigma _{3} \\
\gamma _{1}^{3}\rightarrow \sigma _{1}\ \ \ \ \gamma _{2}^{3}\rightarrow
\sigma _{2}\ \ \ \ \gamma _{3}^{3}\rightarrow \sigma _{3}%
\end{array}%
\right.
\end{equation}%
\begin{equation}
D^{(5)}:\left\{
\begin{array}{c}
\gamma _{1}^{1}\rightarrow \sigma _{1}\ \ \ \ \gamma _{2}^{1}\rightarrow
\sigma _{2}\ \ \ \ \gamma _{3}^{1}\rightarrow \sigma _{3} \\
\gamma _{1}^{2}\rightarrow -\sigma _{1}\ \ \ \ \gamma _{2}^{2}\rightarrow
-\sigma _{2}\ \ \ \ \gamma _{3}^{2}\rightarrow -\sigma _{3} \\
\gamma _{1}^{3}\rightarrow -\sigma _{1}\ \ \ \ \gamma _{2}^{3}\rightarrow
-\sigma _{2}\ \ \ \ \gamma _{3}^{3}\rightarrow -\sigma _{3}%
\end{array}%
\right.
\end{equation}%
\begin{equation}
D^{(6)}:\left\{
\begin{array}{c}
\gamma _{1}^{1}\rightarrow -\sigma _{1}\ \ \ \ \gamma _{2}^{1}\rightarrow
-\sigma _{2}\ \ \ \ \gamma _{3}^{1}\rightarrow -\sigma _{3} \\
\gamma _{1}^{2}\rightarrow -\sigma _{1}\ \ \ \ \gamma _{2}^{2}\rightarrow
-\sigma _{2}\ \ \ \ \gamma _{3}^{2}\rightarrow -\sigma _{3} \\
\gamma _{1}^{3}\rightarrow \sigma _{1}\ \ \ \ \gamma _{2}^{3}\rightarrow
\sigma _{2}\ \ \ \ \gamma _{3}^{3}\rightarrow \sigma _{3}%
\end{array}%
\right.
\end{equation}%
\begin{equation}
D^{(7)}:\left\{
\begin{array}{c}
\gamma _{1}^{1}\rightarrow -\sigma _{1}\ \ \ \ \gamma _{2}^{1}\rightarrow
-\sigma _{2}\ \ \ \ \gamma _{3}^{1}\rightarrow -\sigma _{3} \\
\gamma _{1}^{2}\rightarrow \sigma _{1}\ \ \ \ \gamma _{2}^{2}\rightarrow
\sigma _{2}\ \ \ \ \gamma _{3}^{2}\rightarrow \sigma _{3} \\
\gamma _{1}^{3}\rightarrow -\sigma _{1}\ \ \ \ \gamma _{2}^{3}\rightarrow
-\sigma _{2}\ \ \ \ \gamma _{3}^{3}\rightarrow -\sigma _{3}%
\end{array}%
\right.
\end{equation}%
\begin{equation}
D^{(8)}:\left\{
\begin{array}{c}
\gamma _{1}^{1}\rightarrow -\sigma _{1}\ \ \ \ \gamma _{2}^{1}\rightarrow
-\sigma _{2}\ \ \ \ \gamma _{3}^{1}\rightarrow -\sigma _{3} \\
\gamma _{1}^{2}\rightarrow -\sigma _{1}\ \ \ \ \gamma _{2}^{2}\rightarrow
-\sigma _{2}\ \ \ \ \gamma _{3}^{2}\rightarrow -\sigma _{3} \\
\gamma _{1}^{3}\rightarrow -\sigma _{1}\ \ \ \ \gamma _{2}^{3}\rightarrow
-\sigma _{2}\ \ \ \ \gamma _{3}^{3}\rightarrow -\sigma _{3}%
\end{array}%
\right.
\end{equation}%
We can verify that these representations are nonequivalent analyzing the realizations of basis elements
\begin{equation}
\left\{
\begin{array}{c}
i\gamma _{1}\gamma _{2}\gamma _{3}\otimes 1\otimes 1 \\
1\otimes i\gamma _{1}\gamma _{2}\gamma _{3}\otimes 1 \\
1\otimes 1\otimes i\gamma _{1}\gamma _{2}\gamma _{3}%
\end{array}%
\right.
\end{equation}

We will show that obtaining one-dimensional nonequivalent irreducible representations for subalgebras of $Cl_{3}
\otimes Cl_{3}\otimes \cdots\otimes Cl_{3}$ may allow the construction of decoherence-free subspaces, that is, insensitive subspaces to the action of a certain noise.

Consider as an example the case of three qubits in our formulation. Let $\Gamma_{1}$  be an error operator \cite{lidar1}  given by
\begin{eqnarray*}
\Gamma _{1}&=&k_{1,1}1\otimes 1\otimes 1+k_{1,2}\gamma _{3}\otimes \gamma
_{3}\otimes 1+ \\
&&k_{1,3}1\otimes \gamma _{3}\otimes \gamma _{3}+ k_{1,4}\gamma
_{3}\otimes 1\otimes \gamma _{3},
\end{eqnarray*}%
where $k_{i,j} \in \mathbb{C}$. This error operator is the most general element of a subalgebra $A_{1}$ of $Cl_{3} \otimes Cl_{3}\otimes C_{3}$
generated by the elements $\gamma_{3}\otimes \gamma_{3}\otimes 1$ and $1\otimes \gamma _{3}\otimes \gamma _{3}$.
This subalgebra has four basis elements that commute and therefore, by Theorem \ref{theo17}, we have 4 nonequivalent irreducible representations.
The dimensions of the representations, by Theorem \ref{theo18}, is given by the equation
\begin{equation}
4=n_{1}^{2}+n_{2}^{2}+n_{3}^{2}+n_{4}^{2},
\end{equation}%
The only solution is $n=1$. So we have four one-dimensional irreducible representations. We can find invariant subspaces from the projectors algebra, previously defined by
\begin{equation}
P^{(j)}=\frac{d_{j}}{d}\underset{y}{\sum }\chi ^{(j)\ast }(a_{y})P_{a_{y}},
\end{equation}%
where $d$ is the dimension of the algebra, $d_{j}$ is the dimension of the $j$th representation, $\chi$ is the character of the elements
of algebra, the set $a_{y}$ are basis elements of the algebra and $P_{a_{y}}$ are operators corresponding to the basis elements $a_{y}$.
Now, to get them explicitly, it is also necessary to know the character table of algebra. This, in turn, can be determined by using the following orthogonality relation
\begin{equation}
\underset{y}{\sum }\chi ^{(\alpha )\ast }(a_{y})\chi ^{(\beta
)}(a_{y})=d\delta _{\alpha ,\beta },
\end{equation}
previously deduced (see Theorem \ref{theo24}). For the above example
\begin{equation}
\left\vert \chi (a_{1})\right\vert ^{2}+\left\vert \chi (a_{2})\right\vert
^{2}+\left\vert \chi (a_{3})\right\vert ^{2}+\left\vert \chi
(a_{4})\right\vert ^{2}=4.
\end{equation}%
whose solution is given in table
\begin{equation*}
\begin{tabular}{|c|c|c|c|c|}
\hline
& $1\otimes 1\otimes 1$ & $\gamma _{3}\otimes \gamma _{3}\otimes 1$ & $%
1\otimes \gamma _{3}\otimes \gamma _{3}$ & $\gamma _{3}\otimes 1\otimes
\gamma _{3}$ \\ \hline
$R^{(1)}$ & $1$ & $1$ & $1$ & $1$ \\ \hline
$R^{(2)}$ & $1$ & $-1$ & $1$ & $-1$ \\ \hline
$R^{(3)}$ & $1$ & $-1$ & $-1$ & $1$ \\ \hline
$R^{(4)}$ & $1$ & $1$ & $-1$ & $-1$ \\ \hline
\end{tabular}
\end{equation*}%
where $R^{(1)}, R^{(2)}, R^{(3)}$ and $R^{(4)}$ are one-dimensional nonequivalent irreducible representations. Then the projectors are given by
\begin{eqnarray*}
P_{1}^{1} &=&1/4(1\otimes 1\otimes 1+\gamma _{3}\otimes \gamma _{3}\otimes
1+1\otimes \gamma _{3}\otimes \gamma _{3}+ \\
&&\gamma _{3}\otimes 1\otimes \gamma_{3}); 
\end{eqnarray*}
\begin{eqnarray*}
P_{1}^{2} &=&1/4(1\otimes 1\otimes 1-\gamma _{3}\otimes \gamma _{3}\otimes
1+1\otimes \gamma _{3}\otimes \gamma _{3}- \\
&&\gamma _{3}\otimes 1\otimes \gamma_{3});
\end{eqnarray*}
\begin{eqnarray*}
P_{1}^{3} &=&1/4(1\otimes 1\otimes 1-\gamma _{3}\otimes \gamma _{3}\otimes
1-1\otimes \gamma _{3}\otimes \\
&&\gamma _{3}+\gamma _{3}\otimes 1\otimes \gamma_{3});
\end{eqnarray*}
\begin{eqnarray*}
P_{1}^{4} &=&1/4(1\otimes 1\otimes 1+\gamma _{3}\otimes \gamma _{3}\otimes
1-1\otimes \gamma _{3}\otimes \\
&&\gamma _{3}-\gamma _{3}\otimes 1\otimes \gamma_{3}).
\end{eqnarray*}%
where in $P_{i}^{j}$, the superscript $j$ refers to the irreducible representation and the subscript $i$ indicates the error operator
$\Gamma_{1}$ belonging to the algebra $A_{1}$. Making an analysis in terms of projectors acting on state $(\gamma _{3}\otimes \gamma _{3}\otimes \gamma _{3})(\varepsilon
_{1}\otimes \varepsilon _{1}\otimes \varepsilon _{1}) +(\gamma _{1}\otimes \gamma _{1}\otimes \gamma _{1})(\varepsilon
_{1}\otimes \varepsilon _{1}\otimes \varepsilon _{1})$ taken
arbitrarily, we have by using the relations (\ref{clifford1}) and (\ref{clifford2})
\begin{eqnarray*}
\psi _{1}^{1}&=&P_{1}^{1}[(\gamma _{3}\otimes \gamma _{3}\otimes \gamma _{3})(\varepsilon
_{1}\otimes \varepsilon _{1}\otimes \varepsilon _{1}) +(\gamma _{1}\otimes \gamma _{1}\otimes \gamma _{1})(\varepsilon
_{1}\otimes \varepsilon _{1}\otimes \varepsilon _{1})] \\
&=&(\gamma _{3}\otimes \gamma _{3} \otimes \gamma _{3})(\varepsilon_{1}\otimes \varepsilon _{1}\otimes \varepsilon _{1}) +(\gamma _{1}\otimes \gamma _{1}\otimes \gamma _{1})(\varepsilon_{1}\otimes \varepsilon _{1}\otimes \varepsilon _{1});
\end{eqnarray*}
\begin{eqnarray*}
\psi _{1}^{2}=P_{1}^{2}[(\gamma _{3}\otimes \gamma _{3}\otimes \gamma _{3})(\varepsilon
_{1}\otimes \varepsilon _{1}\otimes \varepsilon _{1})+(\gamma _{1}\otimes \gamma _{1}\otimes \gamma _{1})(\varepsilon
_{1}\otimes \varepsilon _{1}\otimes \varepsilon _{1})] &=&0;
\end{eqnarray*}
\begin{eqnarray*}
\psi _{1}^{3}=P_{1}^{3}[(\gamma _{3}\otimes \gamma _{3}\otimes \gamma _{3})(\varepsilon
_{1}\otimes \varepsilon _{1}\otimes \varepsilon _{1})+(\gamma _{1}\otimes \gamma _{1}\otimes \gamma _{1})(\varepsilon
_{1}\otimes \varepsilon _{1}\otimes \varepsilon _{1})] &=&0;
\end{eqnarray*}
\begin{eqnarray*}
\psi _{1}^{4}=P_{1}^{4}[(\gamma _{3}\otimes \gamma _{3}\otimes \gamma _{3})(\varepsilon
_{1}\otimes \varepsilon _{1}\otimes \varepsilon _{1})+ (\gamma _{1}\otimes \gamma _{1}\otimes \gamma _{1})(\varepsilon
_{1}\otimes \varepsilon _{1}\otimes \varepsilon _{1})]&=&0,
\end{eqnarray*}%
where in $\psi_{i}^{j}$, the superscript $j$ refers to the irreducible representation and the subscript $i$ indicates the error operator $\Gamma_{1}$ belonging to the algebra $A_{1}$. Consequently,
\begin{eqnarray*}
\Gamma _{1} \psi _{1}^{1} &=&\Gamma _{1}[(\gamma _{3}\otimes
\gamma _{3}\otimes \gamma _{3})(\varepsilon _{1}\otimes \varepsilon
_{1}\otimes \varepsilon _{1})+ \\
&&(\gamma _{1}\otimes \gamma _{1}\otimes \gamma _{1})(\varepsilon
_{1}\otimes \varepsilon _{1}\otimes \varepsilon _{1})] \\
&=&(k_{1,1}+k_{1,2}+k_{1,3}+k_{1,4})[(\gamma _{3} \otimes \gamma _{3}\otimes \gamma _{3})(\varepsilon _{1}\otimes \varepsilon _{1}\otimes \varepsilon_{1})+ \\
&&(\gamma _{1}\otimes \gamma _{1}\otimes \gamma _{1})(\varepsilon
_{1}\otimes \varepsilon _{1}\otimes \varepsilon _{1})] \\
&=&(k_{1,1}+k_{1,2}+k_{1,3}+k_{1,4})\psi _{1}^{1},
\end{eqnarray*}%
and the subspace $\psi _{1}^{1}=(\gamma _{3}\otimes \gamma
_{3}\otimes \gamma _{3})(\varepsilon _{1}\otimes \varepsilon _{1}\otimes
\varepsilon _{1})+(\gamma _{1}\otimes \gamma _{1}\otimes \gamma
_{1})(\varepsilon _{1}\otimes \varepsilon _{1}\otimes \varepsilon _{1})$ of the minimal left ideal of $Cl_{3}\otimes Cl_{3}\otimes Cl_{3}$, corresponding to the $GHZ$ state \cite{Dur} in terms of the Hilbert space, it is insensitive to the effect of noise $\Gamma _{1}$, i.e. this space is a decoherence-free subspace be the effect of noise $\Gamma _{1}$.

Now consider the effect of noise on a system of four qubits:
\begin{eqnarray*}
\Gamma _{2}&=&k_{2,1}1\otimes 1\otimes 1\otimes 1+k_{2,2}\gamma _{1}\otimes
\gamma _{1}\otimes 1\otimes \\
&&1+k_{2,3}1\otimes 1\otimes \gamma _{1}\otimes \gamma _{1}+k_{2,4}\gamma _{1}\otimes\\
&&\gamma _{1}\otimes \gamma _{1}\otimes \gamma _{1}
\end{eqnarray*}%
In particular, by our formulation, the projection operators are given by
\begin{eqnarray*}
P_{2}^{1} &=&1/4(1\otimes 1\otimes 1\otimes 1+\gamma _{1}\otimes \gamma
_{1}\otimes 1\otimes 1+\\
&&1\otimes 1\otimes \gamma _{1}\otimes \gamma
_{1}+\gamma _{1}\otimes \gamma _{1}\otimes \gamma _{1}\otimes \gamma _{1}),
\end{eqnarray*}
\begin{eqnarray*}
P_{2}^{2} &=&1/4(1\otimes 1\otimes 1\otimes 1+\gamma _{1}\otimes \gamma
_{1}\otimes 1\otimes 1+\\
&&1\otimes 1\otimes \gamma _{1}\otimes \gamma
_{1}+\gamma _{1}\otimes \gamma _{1}\otimes \gamma _{1}\otimes \gamma _{1}),
\end{eqnarray*}
\begin{eqnarray*}
P_{2}^{3} &=&1/4(1\otimes 1\otimes 1\otimes 1+\gamma _{1}\otimes \gamma
_{1}\otimes 1\otimes 1+ \\
&&1\otimes 1\otimes \gamma _{1}\otimes \gamma
_{1}+\gamma _{1}\otimes \gamma _{1}\otimes \gamma _{1}\otimes \gamma _{1}),
\end{eqnarray*}
\begin{eqnarray*}
P_{2}^{4} &=&1/4(1\otimes 1\otimes 1\otimes 1+\gamma _{1}\otimes \gamma
_{1}\otimes 1\otimes 1+ \\
&&1\otimes 1\otimes \gamma _{1}\otimes \gamma
_{1}+\gamma _{1}\otimes \gamma _{1}\otimes \gamma _{1}\otimes \gamma _{1}).
\end{eqnarray*}%
Applying these projectors in the state $(\gamma _{3}\otimes \gamma _{3}\otimes
\gamma _{3}\otimes \gamma _{3})(\varepsilon _{1}\otimes \varepsilon
_{1}\otimes \varepsilon _{1}\otimes \varepsilon _{1})$, we obtain
\begin{eqnarray*}
\psi _{2}^{1} &=&P_{2}^{1}(\gamma _{3}\otimes \gamma _{3}\otimes \gamma
_{3}\otimes \gamma _{3})(\varepsilon _{1}\otimes \varepsilon _{1}\otimes
\varepsilon _{1}\otimes \varepsilon _{1}) \\
&=&[(\gamma _{3}\otimes \gamma _{3}\otimes \gamma _{3}\otimes \gamma
_{3})+(\gamma _{1}\otimes \gamma _{1}\otimes \gamma _{3}\otimes \gamma _{3})
\\
&&+(\gamma _{3}\otimes \gamma _{3}\otimes \gamma _{1}\otimes \gamma
_{1})+ (\gamma _{1}\otimes \gamma _{1}\otimes \gamma _{1}\otimes \gamma
_{1})]\\
&&(\varepsilon _{1}\otimes \varepsilon _{1}\otimes \varepsilon
_{1}\otimes \varepsilon _{1}),
\end{eqnarray*}%
\begin{eqnarray*}
\psi _{2}^{2} &=&P_{2}^{2}(\gamma _{3}\otimes \gamma _{3}\otimes \gamma
_{3}\otimes \gamma _{3})(\varepsilon _{1}\otimes \varepsilon _{1}\otimes
\varepsilon _{1}\otimes \varepsilon _{1}) \\
&=&[(\gamma _{3}\otimes \gamma _{3}\otimes \gamma _{3}\otimes \gamma
_{3})+(\gamma _{1}\otimes \gamma _{1}\otimes \gamma _{3}\otimes \gamma _{3})
\\
&&-(\gamma _{3}\otimes \gamma _{3}\otimes \gamma _{1}\otimes \gamma
_{1})-(\gamma _{1}\otimes \gamma _{1}\otimes \gamma _{1}\otimes \gamma
_{1})]\\
&&(\varepsilon _{1}\otimes \varepsilon _{1}\otimes \varepsilon
_{1}\otimes \varepsilon _{1}),
\end{eqnarray*}%
\begin{eqnarray*}
\psi _{2}^{3} &=&P_{2}^{3}(\gamma _{3}\otimes \gamma _{3}\otimes \gamma
_{3}\otimes \gamma _{3})(\varepsilon _{1}\otimes \varepsilon _{1}\otimes
\varepsilon _{1}\otimes \varepsilon _{1}) \\
&=&[(\gamma _{3}\otimes \gamma _{3}\otimes \gamma _{3}\otimes \gamma
_{3})-(\gamma _{1}\otimes \gamma _{1}\otimes \gamma _{3}\otimes \gamma _{3})
\\
&&+(\gamma _{3}\otimes \gamma _{3}\otimes \gamma _{1}\otimes \gamma
_{1})-(\gamma _{1}\otimes \gamma _{1}\otimes \gamma _{1}\otimes \gamma
_{1})]\\
&&(\varepsilon _{1}\otimes \varepsilon _{1}\otimes \varepsilon
_{1}\otimes \varepsilon _{1}),
\end{eqnarray*}%
\begin{eqnarray*}
\psi _{2}^{4} &=&P_{2}^{4}(\gamma _{3}\otimes \gamma _{3}\otimes \gamma
_{3}\otimes \gamma _{3})(\varepsilon _{1}\otimes \varepsilon _{1}\otimes
\varepsilon _{1}\otimes \varepsilon _{1}) \\
&=&[(\gamma _{3}\otimes \gamma _{3}\otimes \gamma _{3}\otimes \gamma
_{3})-(\gamma _{1}\otimes \gamma _{1}\otimes \gamma _{3}\otimes \gamma _{3})
\\
&&-(\gamma _{3}\otimes \gamma _{3}\otimes \gamma _{1}\otimes \gamma
_{1})+(\gamma _{1}\otimes \gamma _{1}\otimes \gamma _{1}\otimes \gamma
_{1})]\\
&&(\varepsilon _{1}\otimes \varepsilon _{1}\otimes \varepsilon
_{1}\otimes \varepsilon _{1}).
\end{eqnarray*}%
We then have 4 invariant subspaces given by
\begin{eqnarray*}
\Gamma _{2}\psi _{2}^{1} &=&k_{2,1}\psi _{2}^{1}+k_{2,2}\psi
_{2}^{1}+k_{2,3}\psi _{2}^{1}+k_{2,4}\psi
_{2}^{1}=(k_{2,1}+\\
&&k_{2,2}+k_{2,3}+k_{2,4})\psi _{2}^{1}, \\
\Gamma _{2}\psi _{2}^{2} &=&k_{2,1}\psi _{2}^{2}+k_{2,2}\psi
_{2}^{2}+k_{2,3}\psi _{2}^{2}+k_{2,4}\psi
_{2}^{2}=(k_{2,1}+\\
&&k_{2,2}+k_{2,3}+k_{2,4})\psi _{2}^{2}, \\
\Gamma _{2}\psi _{2}^{3} &=&k_{2,1}\psi _{2}^{2}+k_{2,2}\psi
_{2}^{1}+k_{2,3}\psi _{2}^{1}+k_{2,4}\psi
_{2}^{1}=(k_{2,1}+\\
&&k_{2,2}+k_{2,3}+k_{2,4})\psi _{2}^{3}, \\
\Gamma _{2}\psi _{2}^{4} &=&k_{2,1}\psi _{2}^{2}+k_{2,2}\psi
_{2}^{1}+k_{2,3}\psi _{2}^{1}+k_{2,4}\psi
_{2}^{1}=(k_{2,1}+\\
&&k_{2,2}+k_{2,3}+k_{2,4})\psi _{2}^{4}.
\end{eqnarray*}

Thus the subspace $\psi _{2}^{1}+\psi _{2}^{2}+\psi _{2}^{3}+\psi_{2}^{4}$ is decoherence-free by the action of noise. There is a correspondence with the results of Refs. \cite{lidar1,lidar2}. Finally, consider an example involving bivectors; the noise is given by
\begin{eqnarray}
\Gamma _{3} &=&k_{3,1}1\otimes 1\otimes 1\otimes 1+k_{3,2}\gamma _{1}\otimes
\gamma _{1}\otimes \gamma _{1}\otimes \notag \\
&&\gamma _{1}+k_{3,3}\gamma _{2}\otimes \gamma _{2}\otimes \gamma _{2}\otimes \gamma
_{2}+k_{3,4}\gamma _{1}\gamma _{2}  \nonumber\\
&&\otimes \gamma _{1}\gamma _{2}\otimes\gamma _{1}\gamma _{2}\otimes \gamma _{1}\gamma _{2}.
\end{eqnarray}%
Note that
\begin{equation}
\gamma _{1}(\gamma _{1}\gamma _{2})=-(\gamma _{1}\gamma _{2})\gamma _{1}.
\end{equation}%
However, the tensor product
\begin{equation}
(\gamma _{1}\otimes \gamma _{1})(\gamma _{1}\gamma _{2}\otimes \gamma
_{1}\gamma _{2})=(\gamma _{1}\gamma _{2}\otimes \gamma _{1}\gamma
_{2})(\gamma _{1}\otimes \gamma _{1}),
\end{equation}%
commutes. This always occurs for a pair number of factors, which corresponds to the above case, where there are four factors.
The same goes for all other elements, so we have four elements that commute between themselves resulting in four one-dimensional
irreducible representations. The projectors are given by
\begin{eqnarray*}
P_{3}^{1} &=&1/4(1\otimes 1\otimes 1\otimes 1+\gamma _{1}\otimes \gamma
_{1}\otimes \gamma _{1}\otimes \gamma _{1}  \notag \\
&&+\gamma _{2}\otimes \gamma _{2}\otimes \gamma _{2}\otimes \gamma
_{2}+\gamma _{1}\gamma _{2}\otimes \gamma _{1}\gamma _{2} \otimes \\
&&\gamma_{1}\gamma _{2}\otimes \gamma _{1}\gamma _{2}),  \notag \\
P_{3}^{2} &=&1/4(1\otimes 1\otimes 1\otimes 1+\gamma _{1}\otimes \gamma
_{1}\otimes \gamma _{1}\otimes \gamma _{1}  \notag \\
&&-\gamma _{2}\otimes \gamma _{2}\otimes \gamma _{2}\otimes \gamma
_{2}-\gamma _{1}\gamma _{2}\otimes \gamma _{1}\gamma _{2}\otimes \\
&&\gamma_{1}\gamma _{2}\otimes \gamma _{1}\gamma _{2}),  \notag \\
P_{3}^{3} &=&1/4(1\otimes 1\otimes 1\otimes 1-\gamma _{1}\otimes \gamma
_{1}\otimes \gamma _{1}\otimes \gamma _{1}  \notag \\
&&+\gamma _{2}\otimes \gamma _{2}\otimes \gamma _{2}\otimes \gamma
_{2}-\gamma _{1}\gamma _{2}\otimes \gamma _{1}\gamma _{2}\otimes \\
&&\gamma_{1}\gamma _{2}\otimes \gamma _{1}\gamma _{2}),  \notag \\
P_{3}^{4} &=&1/4(1\otimes 1\otimes 1\otimes 1-\gamma _{1}\otimes \gamma
_{1}\otimes \gamma _{1}\otimes \gamma _{1}  \notag \\
&&-\gamma _{2}\otimes \gamma _{2}\otimes \gamma _{2}\otimes \gamma
_{2}+\gamma _{1}\gamma _{2}\otimes \gamma _{1}\gamma _{2}\otimes \\
&&\gamma_{1}\gamma _{2}\otimes \gamma _{1}\gamma _{2}).
\end{eqnarray*}%
The invariant subspaces are given by
\begin{eqnarray*}
\psi _{3}^{1} &=&P_{3}^{1}(\gamma _{3}\otimes \gamma _{3}\otimes \gamma
_{3}\otimes \gamma _{3})(\varepsilon _{1}\otimes \varepsilon _{1}\otimes
\varepsilon _{1}\otimes \varepsilon _{1}) \\
&=&2[(\gamma _{3}\otimes \gamma _{3}\otimes \gamma _{3}\otimes \gamma
_{3})(\varepsilon _{1}\otimes \varepsilon _{1}\otimes \varepsilon
_{1}\otimes \varepsilon _{1}) \\
&&+(\gamma _{1}\otimes \gamma _{1}\otimes \gamma _{1}\otimes \gamma
_{1})(\varepsilon _{1}\otimes \varepsilon _{1}\otimes \varepsilon
_{1}\otimes \varepsilon _{1}), \\
\psi _{3}^{2} &=&P_{3}^{2}(\gamma _{3}\otimes \gamma _{3}\otimes \gamma
_{3}\otimes \gamma _{3})(\varepsilon _{1}\otimes \varepsilon _{1}\otimes
\varepsilon _{1}\otimes \varepsilon _{1}) \\
&=&0, \\
\psi _{3}^{3} &=&P_{3}^{3}(\gamma _{3}\otimes \gamma _{3}\otimes \gamma
_{3}\otimes \gamma _{3})(\varepsilon _{1}\otimes \varepsilon _{1}\otimes
\varepsilon _{1}\otimes \varepsilon _{1}) \\
&=&0, \\
\psi _{3}^{4} &=&P_{3}^{4}(\gamma _{3}\otimes \gamma _{3}\otimes \gamma
_{3}\otimes \gamma _{3})(\varepsilon _{1}\otimes \varepsilon _{1}\otimes
\varepsilon _{1}\otimes \varepsilon _{1}) \\
&=&2[(\gamma _{3}\otimes \gamma _{3}\otimes \gamma _{3}\otimes \gamma
_{3})(\varepsilon _{1}\otimes \varepsilon _{1}\otimes \varepsilon
_{1}\otimes \varepsilon _{1}) \\
&&-(\gamma _{1}\otimes \gamma _{1}\otimes \gamma _{1}\otimes \gamma
_{1})(\varepsilon _{1}\otimes \varepsilon _{1}\otimes \varepsilon
_{1}\otimes \varepsilon _{1})].
\end{eqnarray*}%
The effect of noise on $\psi_{i}^{j}$ results in
\begin{eqnarray*}
\Gamma _{3}\psi _{3}^{1} &=&\alpha _{3,1}\psi _{3}^{1}+\alpha _{3,2}\psi
_{3}^{1}+\alpha _{3,3}\psi _{3}^{1}+\alpha _{3,4}\psi _{3}^{1}\\
&=&(\alpha_{3,1}+\alpha _{3,2}+\alpha _{3,3}+\alpha _{3,4})\psi _{3}^{1},
\end{eqnarray*}
\begin{eqnarray*}
\Gamma _{3}\psi _{3}^{4} &=&\alpha _{3,1}\psi _{3}^{4}+\alpha _{3,2}\psi
_{3}^{4}+\alpha _{3,3}\psi _{3}^{4}+\alpha _{3,4}\psi _{3}^{4} \\
&=&(\alpha_{3,1}+\alpha _{3,2}+\alpha _{3,3}+\alpha _{3,4})\psi _{3}^{4}.
\end{eqnarray*}
Thus $\psi _{3}^{1}+\psi _{3}^{4}$ is a decoherence-free subspace on the effect of noise $\Gamma_{3}$.

\section{Conclusions} \label{sec:conclu}
Algebraic methods have been useful in the developed of the quantum mechanics since its beginnings. In this paper we present a general approach that enables the construction of the decoherence-free subspaces in a purely algebraic scenario. For this purpose, we first generalize some theorems about representations of semisimple algebras due Pauli and Artin. Then we derive orthogonality theorems for irreducible representations of algebras. Based on these mathematical results, we developed a scheme using the tensor product of Clifford algebras and minimals left ideals to construct the decoherence-free subspaces. An advantage of this formalism is that it provides a systematic method for the construction of DFSs in the composite systems within the same algebraic structure. Another advantage of this formalism is that all elements (states and operators) can be described in terms of generators of the tensor product of algebras. Also, since we are describing the systems in an algebraic way, different representations may be derived conveniently.

\end{document}